\documentclass{article}
\pdfoutput=1
\usepackage{etex} % Fixes out-of-registers error. Probably some bad package.
\reserveinserts{28}

\usepackage{xcolor}
\definecolor{amazing}{RGB}{254,67,101}

\usepackage{bbm}
\usepackage{mathrsfs}
\usepackage{latexsym}
\usepackage{paralist}
\usepackage{suffix}

\usepackage{xpunctuate}
\usepackage[all,british]{foreign} % Typesets abbreviations like e.g.
\usepackage{xfrac} % Nice in-text fractions
\usepackage[style=english,english=british]{csquotes}

\usepackage{amsmath, amsthm, amssymb, amsfonts} % AMS packages
\usepackage{thmtools, thm-restate}
\usepackage{mathtools}

\usepackage[full,small]{complexity}

\usepackage[final,notref,color]{showkeys} % Use option 'final' to remove

\usepackage{xargs}
\usepackage{afterpage} % Provides macro to execute commands 
\PassOptionsToPackage{dvipsnames,usenames,table}{xcolor}
\usepackage{tikz}
\usetikzlibrary{calc}
\usepackage{booktabs} % Better tables
\usepackage{longtable}
\usepackage{float}
\usepackage{wrapfig}
\usepackage{floatflt}
\usepackage{framed}
\usepackage{ctable}
\usepackage{dcolumn}
\usepackage{adjustbox}

\usepackage{enumerate}
\usepackage{enumitem}
\usepackage{xspace}
\usepackage{xifthen}
\usepackage{fixltx2e}
\usepackage{url}
\usepackage{scrtime}

\usepackage[longend,vlined]{algorithm2e}
\usepackage{xkvltxp}
\usepackage[draft,english,silent]{fixme}
\newcommand{\todo}[1]{\fxfatal{\color{red}#1}}

\usepackage{index}

\makeindex
\newindex{sym}{sym.idx}{sym.ind}{Symbol index}
\newindex{prob}{prob.idx}{prob.ind}{Problem index}

\newcommand\mathindex[1]{\index[sym]{\ensuremath{{#1}}}}
\WithSuffix\newcommand\mathindex*[1]{\index*[sym]{\ensuremath{{#1}}}}

\makeatletter
\newcommand{\raisemath}[1]{\mathpalette{\raisem@th{#1}}}
\newcommand{\raisem@th}[3]{\raisebox{#1}{$#2#3$}}
\makeatother

\newcommand\Problem[2][]{%
    \index[prob]{#2@\textsc{#2}}\textsc{#2}\xspace%
}

\newclass{\paraNP}{paraNP}

\def\MSOii{\textsc{MSO${}_2$}\xspace}

\newcommand\restr[2]{{% Restriction of functions, set families
  \left.\kern-\nulldelimiterspace % automatically resize the bar with \right
  #1 % the function
  \vphantom{\big|} % pretend it's a little taller at normal size
  \right|_{#2} % this is the delimiter
  }}

\def\half{{\sfrac{1}{2}}} % Needs second block to work as subscript
\def\fivehalf{{\sfrac{5}{2}}} % Needs second block to work as subscript

\def\sminor^#1{
    \preccurlyeq_{\mathrlap{\mathsf{m}}}^{#1}
} % Shallow minor relation
\def\ssminor^#1{%
    \mathbin{\dot\preccurlyeq_{{\mathrlap{\mathsf{m}}}}^{#1}}\,
} % Stable shallow minor relation
\def\stminor^#1{
    \preccurlyeq_{\mathrlap{\mathsf{t}}}^{#1}
} % Shallow minor relation
\def\sstminor^#1{%
    \mathbin{\dot\preccurlyeq_{{\mathrlap{\mathsf{t}}}}^{#1}}\,
} % Stable shallow top. minor relation

\def\grad_#1{\nabla\!_{#1}}
\def\sgrad_#1{{\dot\nabla}\!_{#1}}
\def\topgrad_#1{\widetilde \nabla\!_{#1}}
\def\stopgrad_#1{\dot{\widetilde\nabla}\!_{#1}}
\def\topomega_#1{\widetilde \omega_{#1}}

\def\colnum_#1{ \operatorname{col}_{#1} }
\def\wcolnum_#1{ \operatorname{wcol}_{#1} }
\def\adm_#1{ \operatorname{adm}_{#1} }

\renewcommand{\leq}{\leqslant}

\renewcommand{\geq}{\geqslant}

\renewcommand{\epsilon}{\varepsilon}
\newcommand{\eps}{\epsilon}

\newlength{\convarrowwidth}
\usepackage{stmaryrd}

\newcommand{\widthm}[1]{ \mathbf{#1} } % Width measure style
\DeclareMathOperator{\tw}{ \widthm{tw} }

\DeclareMathOperator{\width}{ \widthm{width} } % Generic width function for decompositions

\def\YYYY{{Y_0 \uplus Y_1 \uplus \cdots \uplus Y_\ell}} % I don't want to type this anymore
\WithSuffix\def\YYYY'{{Y'_0 \uplus Y'_1 \uplus \cdots \uplus Y'_{\ell'}}}

\usepackage{pifont}% http://ctan.org/pkg/pifont
\newcommand{\yaay}{\kern4pt \ding{51} \kern-8pt \ding{51}}%
\theoremstyle{plain}
\newtheorem{lemma}{Lemma}

\newtheorem{theorem}{Theorem}

\newtheorem{proposition}{Proposition}
\newtheorem*{claim}{Claim}

\newtheoremstyle{case}
  {\topsep}   % ABOVESPACE
  {\topsep}   % BELOWSPACE
  {}  % BODYFONT
  {\parindent}       % INDENT (empty value is the same as 0pt)
  {\bfseries} % HEADFONT
  {\normalfont.}         % HEADPUNCT
  {5pt plus 1pt minus 1pt} % HEADSPACE
  {#1 #2: {\normalfont #3}}          % CUSTOM-HEAD-SPEC

\theoremstyle{case}
\newtheorem{case}{Case}

\numberwithin{subcase}{case}
\makeatletter%Make sure case-counter is reset in every environment
\@addtoreset{case}{lemma}
\@addtoreset{case}{theorem}
\@addtoreset{case}{corollary}
\makeatother

\theoremstyle{definition}

\renewcommand*\etal{\xperiodafter{\emph{et~al}}}

\newcommand*\varrule[1][0.4pt]{\leavevmode\leaders\hrule height#1\hfill\kern0pt}

\setlist[1]{labelindent=\parindent,leftmargin=*} 
\setlist{itemsep=0pt}
\setitemize[1]{label={\small\textbullet}}

\newenvironment{tightcenter}
 {\parskip=0pt\par\nopagebreak\centering}
 {\par\noindent\ignorespacesafterend}

\usepackage{ctable}
\newlength{\RoundedBoxWidth}
\newsavebox{\GrayRoundedBox}
\newenvironment{GrayBox}[1]%
   {\setlength{\RoundedBoxWidth}{\textwidth-4.5ex}
    \def\boxheading{#1}
    \begin{lrbox}{\GrayRoundedBox}
       \begin{minipage}{\RoundedBoxWidth}%
   }{%
       \end{minipage}
    \end{lrbox}%
    \begin{tightcenter}%
    \begin{tikzpicture}%
       \node(Text)[draw=black!20,fill=white,rounded corners,%
             inner sep=2ex,text width=\RoundedBoxWidth]%
             {\usebox{\GrayRoundedBox}};
        \coordinate(x) at (current bounding box.north west);
        \node [draw=white,rectangle,inner sep=3pt,anchor=north west,fill=white] 
        at ($(x)+(6pt,.75em)$) {\boxheading};
    \end{tikzpicture}
    \end{tightcenter}\vspace{0pt}%
    \ignorespacesafterend
}    

\newenvironment{problem}[2][]{\noindent\ignorespaces%
                                \FrameSep=6pt%
                                \parindent=0pt%
                \vspace*{-.5em}
                \ifthenelse{\isempty{#1}}{%
                  \begin{GrayBox}{\textsc{#2}}%                
                }{%
                  \begin{GrayBox}{\textsc{#2} parametrised by~{#1}}%  
                }
                \index[prob]{#2@\textsc{#2}}%
                \begin{tabular*}{\textwidth}{@{\hspace{.1em}} >{\itshape} p{1.6cm} p{0.8\textwidth} @{}}%        
            }{
                \end{tabular*}%
                \end{GrayBox}%
                \vspace*{-.5em}
                \ignorespacesafterend
            }  
\newenvironment{problem*}[2][]{\noindent\ignorespaces%
                                \FrameSep=6pt%
                                \parindent=0pt%
                \vspace*{-.5em}
                \ifthenelse{\isempty{#1}}{%
                  \begin{GrayBox}{\textsc{#2}}%                
                }{%
                  \begin{GrayBox}{\textsc{#2} parametrised by~{#1}}%  
                }
                \index[prob]{#2@\textsc{#2}|textbf}%
                \begin{tabular*}{\textwidth}{@{\hspace{.1em}} >{\itshape} p{1.6cm} p{0.8\textwidth} @{}}%        
            }{
                \end{tabular*}%
                \end{GrayBox}%
                \vspace*{-.5em}
                \ignorespacesafterend
            }       

\renewenvironmentx{leftbar}[2][1=0.5pt, 2=5pt]{% 
  \def\FrameCommand{\vrule width #1 \hspace{#2}}\MakeFramed {\advance\hsize-\width \FrameRestore}}%
{\endMakeFramed}

\newlength{\wleft}  \newlength{\wright}

\definecolor{Maroon}{cmyk}{0, 0.87, 0.68, 0.32}
\definecolor{RoyalBlue}{cmyk}{1, 0.50, 0, 0}
\definecolor{Black}{cmyk}{0, 0, 0, 0}
\definecolor{White}{rgb}{1, 1, 1}

\makeatletter
\let\orgdescriptionlabel\descriptionlabel
\def\@savelabel{}
\renewcommand*{\descriptionlabel}[1]{%
  \let\orglabel\label
  \let\label\@gobble
  \phantomsection
  \def\@savelabel{#1}
  \edef\@currentlabel{{\def\hfil{}#1}}% Dirty trick because this adds a \hfil at the end
  \edef\@currentlabelname{#1}%
  \let\label\orglabel
  \orgdescriptionlabel{#1}%
}
\makeatother
\makeatletter
\def\namedlabel#1#2{\begingroup
   \def\@currentlabel{#1}%
   \label{#2}\endgroup
}
\makeatother

\usepackage{hyphenat}
\renewcommand{\th}{%
    \ifmmode% math mode
        ^\mathrm{th}%
    \else%
        \textsuperscript{th}\xspace%
    \fi%
}
\newcommand{\st}{%
    \ifmmode% math mode
        ^\mathrm{st}%
    \else%
        \textsuperscript{st}\xspace%
    \fi%
}
\newcommand{\nd}{%
    \ifmmode% math mode
        ^\mathrm{nd}%
    \else%
        \textsuperscript{nd}\xspace%
    \fi%
}
\newcommand{\rd}{%
    \ifmmode% math mode
        ^\mathrm{rd}%
    \else%
        \textsuperscript{rd}\xspace%
    \fi%
}

\def\Nesetril{Ne\v{s}et\v{r}il\xspace}

\def\Dvorak{Dvo\v{r}\'{a}k\xspace}

\usepackage{environ}% 
\NewEnviron{nofiitable}{\noindent\ignorespaces%
  
  \[\arraycolsep=1.4pt%
  \begin{array}{rcr p{1cm} rcr}
    \BODY
  \end{array}
  \]
}

\newcolumntype{m}{>{$}l<{$}}
\newcolumntype{M}{>{$\displaystyle}l<{$}} 

\newcolumntype{L}{l}  
\newcolumntype{C}{c}  
\newcolumntype{R}{r}  
\newcolumntype{X}{>{\global\let\currentrowstyle\relax}}
\newcolumntype{^}{>{\currentrowstyle}}

\usepackage{authblk}

\title{Being even slightly shallow makes life hard}
\author[1]{Irene~Muzi}
\author[2]{Michael~P.~O'Brien}
\author[2]{Felix~Reidl}
\author[2]{Blair~D.~Sullivan}
\affil[1]{\small University of Rome, ``La Sapienza'', Rome, Italy\\
  \texttt{irene.muzi@gmail.com}}
\affil[2]{\small North Carolina State University, USA\\
  \texttt{(mpobrie3|fjreidl|vbsulliv)@ncsu.edu}}
\def\PosSAT{\Problem{Positive $1$-in-$3$SAT}}

\begin{document}

\maketitle
\begin{abstract}
We study the computational complexity of identifying dense substructures,
namely \emph{$r/2$-shallow topological minors} and \emph{$r$-subdivisions}.
Of particular interest is the case when $r=1$, when these substructures
correspond to very localized relaxations of subgraphs.
Since \Problem{Densest Subgraph} can be solved in polynomial time, we ask whether
these slight relaxations also admit efficient algorithms.

In the following, we provide a negative answer: \Problem{Dense
$r/2$-Shallow Topological Minor} and \Problem{Dense $r$-Subdivsion} are
already \NP-hard for~$r = 1$ in very sparse graphs. Further, they do not
admit algorithms with running time $2^{o(\tw^2)} n^{O(1)}$
when parameterized by the treewidth of the input graph for~$r \geq 2$
unless ETH fails.

\end{abstract}

\section*{Introduction}
Identifying dense substructures\textemdash in particular, subgraphs\textemdash is a frequent task
on graphs in real world-problems. Famously, Karp showed that the extreme
variant of this task, finding a complete subgraph with a certain number of
vertices, is \NP-complete. Asahiro~\etal showed, among other related results,
that the problem remains hard even if we ask for a subgraph
with~$\Theta(k^{1+\eps})$ edges on~$k$ vertices~\cite{kSubgraphComplexity}.
In contrast, the task of finding the
\emph{densest subgraph} without any restriction on its \emph{size}
is efficiently computable using flow-based
methods~\cite{DensestSubgraphFlow,DensestSubgraphFast}.

Many other substructures have played key roles in seminal work in structural
graph theory. For example, minors, topological minors, and immersions were at
the epicenter of Robertson and Seymour's graph minor program---work which
gave rise to a large body of algorithmic advances (see
\eg~\cite{BidimSurveyA,BidimSurveyB,DomsetKernelHMinorFree})
and laid the groundwork for Downey and Fellows' introduction of
parameterized complexity~\cite{DowneyFellows}. More recently,
\emph{shallow} (topological) minors enabled \Nesetril and Ossona de Mendez's
development of a comprehensive theory of sparse
graph classes~\cite{Sparsity}---again spawning a slew of related algorithmic results
(see \eg~\cite{BndExpKernelsJournal,FOBndExp,FONowhereDense,DomsetSparseKernel}
and~\cite{FelixThesis} for an overview). One common thread among
all these ideas is that certain (dense) substructures are \emph{excluded}
in order to gain algorithmic tractability, which conversely means that we would
like to be able to compute the densest occuring substructure in a graph.

It is therefore natural to ask whether finding
dense substructures is efficiently possible. In the case of \emph{minors},
Bodlaender, Wolle, and Koster showed that deciding whether some minor of the
input graph has \emph{degeneracy} larger than~$d$ is \NP-complete~\cite{ContractionDegeneracy}. They also note
that if~$d$ is considered a constant, the problem can be solved in cubic time
using the minor-test by Robertson and Seymour. The same observation holds if
we ask for a minor of \emph{density}~$d$ instead and we can equivalently
state that both problems are in \FPT~when parameterized by the target
degeneracy/density~$d$.

\Dvorak considered the problem of finding an \emph{$r$-shallow minor} of
degeneracy/density at least~$d$~\cite{DvorakThesis}. A graph~$H$
is an~$r$-shallow minor of a graph~$G$ if~$H$ can be obtained from~$G$
by contracting disjoint connected subgraphs of radius at most~$r$.
These substructures offer a natural way of intermediating between the locality
of subgraphs ($0$-shallow minors) and the global nature of
($\infty$-shallow) minor containment. He
proved that both variations are \NP-complete already in graphs of
maximum degree four and~$d \geq 4$ ($d \geq 2$ if degeneracy is the
measure). Accordingly, a parameterization by~$d$ cannot possibly yield
an fpt-algorithm---a sharp contrast to unrestricted minors. \Dvorak
showed that the problem is in \FPT~if parameterized by the treewidth~$\tw$
of the input graph and designed an~$O(4^{\tw^2} n)$ dynamic programming algorithm.

In this paper we focus on $r/2$-shallow \emph{topological} minors
and $r$-\emph{subdivisions}. Recall that a graph~$H$ is an $r$-shallow
topological minor of a graph~$G$ if an $\leq 2r$-subdivision of~$H$
is isomorphic to some subgraph of~$G$.  In particular, for~$r =1$ these
substructures fall between the notion of $0$-shallow minors (subgraphs)
and $1$-shallow minors. The complexity of finding such
substructures is of special interest since the densest $0$-shallow
minor can be computed in polynomial time, while finding the densest
$1$-shallow minor is \NP-complete (even for constant densities).

We show that~\Problem{Dense $r/2$-Shallow Topological Minor}
(\Problem{Dense $r/2$-STM}) and~\Problem{Dense $r$-Subdivision}
(\Problem{Dense $r$-SD}) are \NP-complete already in subcubic
apex-graphs\footnote{That is, a graph in which the removal of a single vertex
results in a subcubic planar graph.} for~$r \geq 1$ via a reduction  from \PosSAT.
Accordingly, a parameterization by the target density~$d$ does not make
these problems fixed-parameter tractable.
The same reduction also implies that neither problem can be solved in
time~$O(2^{o(n)})$ unless the ETH fails.
In other words, finding dense substructures which are just slightly `less
local' than subgraphs seems to be intrinsically difficult.

Following \Dvorak's results, we then consider a parameterization by
treewidth and ask whether an algorithm with running time better than
$O(2^{t^2}n)$ is possible. Surprisingly, we can rule out such an
algorithm already for~\Problem{Dense $1$-Shallow Topological Minor}:
unless the ETH fails, no algorithm with running time~$O(2^{o(t^2)} n)$
can exist.

\section{Preliminaries}
\noindent
For a graph $G$ we use $|G|$ to denote the number of vertices and $\|G\|$
to denote the number of edges in $G$. A graph~$H$ appears as an \emph{$r$-subdivision}
in a graph~$G$ if the graph obtained from~$H$ by subdividing every edge~$r$
times is isomorphic to some subgraph of~$G$. Similarly, $H$ is a
\emph{$r/2$-shallow topological minor}\footnote{The somewhat cumbersome convention of letting an~$r$-shallow
  topological minor contract paths of length~$2r+1$ is convenient in the
  broader context of sparse graph classes (\cf \cite{Sparsity}).
} of~$G$ if a graph obtained from~$H$
by subdividing every edge \emph{up to}~$r$ times is isomorphic to a subgraph
of~$G$. In both cases, the subgraph witnessing the minor is the \emph{model} and we call
those vertices in it that correspond to (subdivided) edges \emph{subdivision
vertices} and all other vertices \emph{nails}.
If $S_{uv}$ is the set of subdivision vertices on a subdivided $uv$ edge, we say $S_{uv}$ is \emph{smoothed} into the $uv$ edge.

\begin{figure}[tb]
  \centering\includegraphics[width=.5\textwidth]{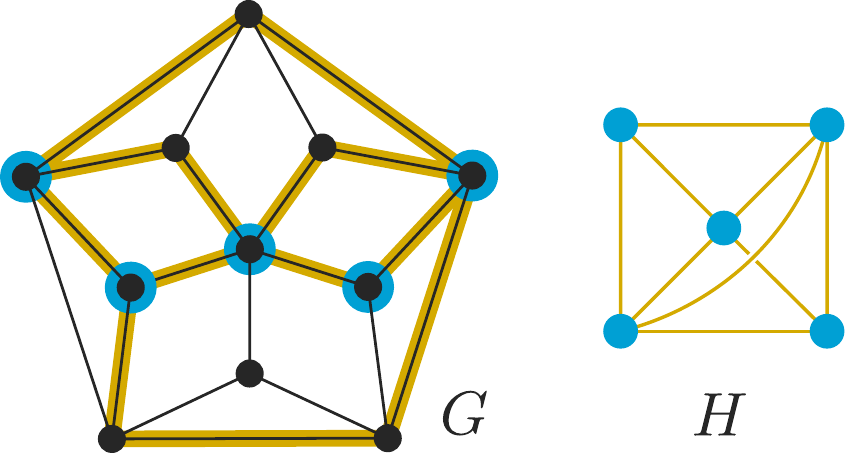}
  \caption{\label{fig:shallowtopminor}%
    The graph~$H$ is a $1$-shallow topological minor of~$G$, as witnessed by
    the model marked with blue nails and golden paths.
  }
\end{figure}

The following two problems are the focus of this paper:

\begin{problem}{Dense $r/2$-Shallow Topological Minor (Dens. $r/2$-STM)}
	Input: & A graph~$G$ and a rational number~$d$. \\
	Question: & Is there an~$r/2$-shallow topological minor~$H$ of~$G$
				with density $\|H\| / |H| \geq d$?
\end{problem} \vspace*{-1em}

\begin{problem}{Dense $r$-Subdivision (Dens. $r$-SD)}
	Input: & A graph~$G$ and a rational number~$d$. \\
	Question: & Is there a graph~$H$ that is contained
						in $G$ as an~$r$-subdivision with density $\|H\| / |H| \geq d$?
\end{problem}

\noindent
The following variant, which we prove to be \NP-complete in
Section~\ref{sec:np-hardness-eth}, might be of independent interest:

\begin{problem}{Dense Bipartite Subdivision}
  Input: & A bipartite graph~$(X,Y,E)$ and a rational number~$d$. \\
  Question: & Are there subsets~$X' \subseteq X, Y' \subseteq Y$ such
              that all vertices in~$X'$ can be smoothed into unique edges
              in~$Y'$ and~$|X'|/|Y'| \geq d$?
\end{problem}

\noindent
Our main tool will be linear reductions from the following \Problem{SAT}-variant:

\begin{problem}{Positive $1$-in-$3$SAT}
    Input:    & A CNF boolean formula $\psi$ with only
    					  positive literals. \\
    Question: & Does~$\phi$ have a satisfying assignment such that
    						each clause contains exactly one true variable?
\end{problem}

\noindent
Mulzer and Rote showed~\cite{Planar1in3Sat} that \PosSAT remains
\NP-hard when restricted to \emph{planar} formulas. A formula~$\phi$ is planar if
the graph obtained from~$\phi$ by creating one vertex for each clause and variable and
connecting a variable-vertex to a clause-vertex if the clause contains said
variable is planar.

Schaefer~\cite{1in3SatHardness} provided a linear reduction from \Problem{3SAT}
to~\Problem{$1$-in-$3$SAT}. We can further easily transform a formula~$\phi$
with negative literals into one with only positive literals as follows:
for each variable~$x$, introduce the variables~$x^+,x^-,a_x,b_x,c_x$.
Replace every occurrence of~$x$ with~$x^+$ and every occurrence of~$\bar x$ with~$x^-$
and add the clauses
\[
	\{x^+, x^-, a_x\}, \{x^+, x^-, b_x\}, \{a_x, b_c, c_x\},
\]
to the formula. It is easy to verify that exactly one of~$x^+,x^-$ must be
true in a $1$-in-$3$ satisfying assignment and that the resulting formula
$\phi'$ has size linear in~$|\phi|$. In conclusion, there exists a linear
reduction from~\Problem{3SAT} to \PosSAT which implies that under ETH,
\PosSAT cannot be solved in time~$2^{o(n)}(n+m)^{O(1)}$, where $n$ is the number of variables and $m$ is the number of clauses.
Using sparsification one can further show that the ETH excludes
algorithms for~\Problem{3SAT} with running time~$2^{o(m)}(n+m)^{O(1)}$
(see \eg the survey by Cygan~\etal~\cite{SETHSurvey}). The above
reduction implies the following lower bound: \looseness-1

\begin{proposition}\label{prop:possat-eth}
  Unless the ETH is false, \PosSAT cannot be solved in time~$2^{o(m)}(n+m)^{O(1)}$.
\end{proposition}

\section{Algorithmic considerations}
\noindent
We start with a basic observation about the problems in question
with the smallest sensible depths of~$r=1$:

\begin{lemma}\label{lem:matching}
  The densest \half-shallow minor or $1$-subdivision on a given set
  of nails can be computed in polynomial time.
\end{lemma}
\begin{proof}
  Assume we are to find the densest $1$-subdivision with nail set~$X$
  in a graph~$G$. We construct an auxilliary bipartite graph~$\hat G$ with
  vertex set~$V(G) \setminus X$ and~${X \choose 2}$ where the
  vertex~$v \in V(G)\setminus X$ is connected to~$xy \in {X \choose 2}$
  iff~$\{x,y\} \subseteq N(v)$ in~$G$, that is, if~$v$ can be contracted
  into the edges~$x,y$. Now simply note that a matching of cardinality~$\ell$
  in~$\hat G$ corresponds to a $1$-subdivision in~$G$ with~$\ell$ subdivisions.
  Finding a maximal matching in~$\hat G$ therefore provides us with the
  densest $1$-subdivision in~$G$ with nail set~$X$. The same proof works
  for~\half-shallow minors if we subdivide all edges existing inside~$X$
  and then construct~$\hat G$.
\end{proof}

\noindent
Consequently, \Problem{Dense $1$-SD} and~\Problem{Dense \half-STM} both admit a
simple~$2^n n^{O(1)}$-algorithm: we guess the nail set~$X$ and apply the
matching construction from Lemma~\ref{lem:matching}. For the same reasons,
both problems are in \XP~when parameterized by the number of nails.
We cannot hope for much better since for~$r=0$ and~$d \sim k^2$ we simply
recover the problem of finding a $k$-clique. Besides being~\W[1]-hard and thus
probably not in \FPT, \Problem{$k$-Clique} further does not admit algorithms
with running time~$f(k) n^{o(k)}$ unless the ETH fails~\cite{ETHCliqueLowerBnd}.

The approach of guessing the nail sets also fails for larger depths: knowing
the nails of a, say, $1$-shallow minors leaves us with the problem of
contracting paths of length two into~$X$, which cannot be represented as a
simple matching problem. The reduction presented in
Section~\ref{sec:tw-lower-bound} proves as a corollary
that~\Problem{Dense $1$-STM} remains \NP-hard when the nail set of the
densest minor is known.

Finally, as we will see in Section~\ref{sec:np-hardness-eth}, both problems
are already \NP-complete for very small densities~$d$, making them \paraNP-complete
under this parameterization. Therefore none of the input variables
will work well as a parameterization, and it is sensible to consider \emph{structural} parameters,
meaning
parameters derived from the input graph. A good contender for such parameters
are \emph{width measures} like tree-, path-, or cliquewidth. Indeed,
we can express the problem of finding a dense shallow minor or a
dense subdivision in \MSOii and apply variants of Courcelle's
theorem to obtain the following:

\begin{proposition}
  \Problem{Dense $r/2$-STM} and \Problem{Dense $r$-SD} are in
  \FPT~when parameterized by the treewidth of the input graph.
\end{proposition}
\begin{proof}
  We can express a model for an $r$-shallow minor in \MSOii as follows: it consists
  of a vertex-set~$W$ and an edge set~$F$, where $F$ induces a set of paths. We
  can further easily express that the paths formed by $F$ are a) of length at most~$r$,
  b) disjoint, and c) have endpoints in~$W$. Lastly, we demand that for every pair~$x,y \in W$
  there exists at most one path in~$F$ that has~$x$ and~$y$ as endpoints.

  From an optimization perspective, we can therefore express the \emph{feasible}
  solutions to \Problem{Dense $r/2$-STM} (and \Problem{Dense $r$-SD} with small
  modifications). In order to express our optimization goal, let us introduce one more
  set of vertices~$C$ with the property that every path induced by~$F$ contains
  at most one vertex from~$C$---for example, we can express in \MSOii that vertices
  of~$C$ are not pairwise reachable via the graph induced by~$F$. With this auxilliary
  set, the density of the resulting minor is at least~$|C|/|W|$ and exactly the
  density if~$C$ is maximial with respect to our choice of~$F$. Accordingly,
  we find that there exists an $r$-shallow topological minor of density at least
  $d$ if $|C| - d|W| \geq 0$. This constraint and the aforementioned \MSOii-description
  of a minor fall within the expressive power of the EMSO-framework introduced
  by Arnborg, Lagergren, and Seese~\cite{EMSO} and we conclude that both
  \Problem{Dense $r/2$-STM} and \Problem{Dense $r$-SD} are fpt when parameterized
  by treewidth.
\end{proof}

\noindent
Furthermore, it is not difficult (albeit tedious) to design a dynamic
programming algorithm that solves \Problem{Dense $r/2$-STM} and
\Problem{Dense $r$-SD} in time~$2^{O(\tw^2)} n$. The quadratic dependence on
the treewidth stems from the fact that we have to keep track of which edges we
have contracted so far and there is no obvious way to circumvent this. The
important question was whether any of the known techniques to reduce the
complexity of connectivity-problems~\cite{CutAndCount,SingleExpTWDet,SingleExpTWLogic}
could be applied here. The answer is, to our surprise, negative as we will discuss in
Section~\ref{sec:tw-lower-bound}.

\section{Hardness results}
\subsection{\NP-hardness and ETH lower bounds}\label{sec:np-hardness-eth}
\noindent
This section will be dedicated to the proof the following theorems which
both follow directly via a linear reduction from~\PosSAT.

\begin{theorem}\label{thm:np-hardness}
	\Problem{Dense $r/2$-STM} and \Problem{Dense $r$-SD} are \NP-hard
	for~$r \geq 1$, even when restricted to graphs that can be turned into
	subcubic planar graphs by deleting a single vertex.
\end{theorem}

\begin{theorem}\label{thm:eth-lowerbound}
	\Problem{Dense $r/2$-STM} and \Problem{Dense $r$-SD} cannot be
	solved in time $2^{o(n)} n^{O(1)}$ on bipartite graphs unless the ETH fails.
\end{theorem}

\noindent
A special case of our result might be of independent interest:

\begin{theorem}\label{thm:dense-bipartite-subdiv}
	\Problem{Dense Bipartite Subdivision} is \NP-hard even on instances
	$\big((X,Y,E),d\big)$ where vertices in~$X$ have degree at most~$3$ and
	$d \geq 3$.
\end{theorem}

\noindent
In the following, we
present two reduction from \PosSAT that depend on the parity
of~$r$. We describe the reduction for~$r \in \{1,2\}$ and then argue
how to modify the construction for arbitrary values of~$r$.
Note that the resulting instances are such that the densest
graph~$H$ that appears as a $r/2$-shallow topological minor
appears, in fact, as an $r$-subdivision and thus the reductions
work for both problems.

%   .88888.        dP       dP
%  d8'   `8b       88       88
%  88     88 .d888b88 .d888b88
%  88     88 88'  `88 88'  `88
%  Y8.   .8P 88.  .88 88.  .88
%   `8888P'  `88888P8 `88888P8
%
%
\paragraph*{Reduction for $r$ odd}

\begin{figure}
	\centering\includegraphics[width=.9\textwidth]{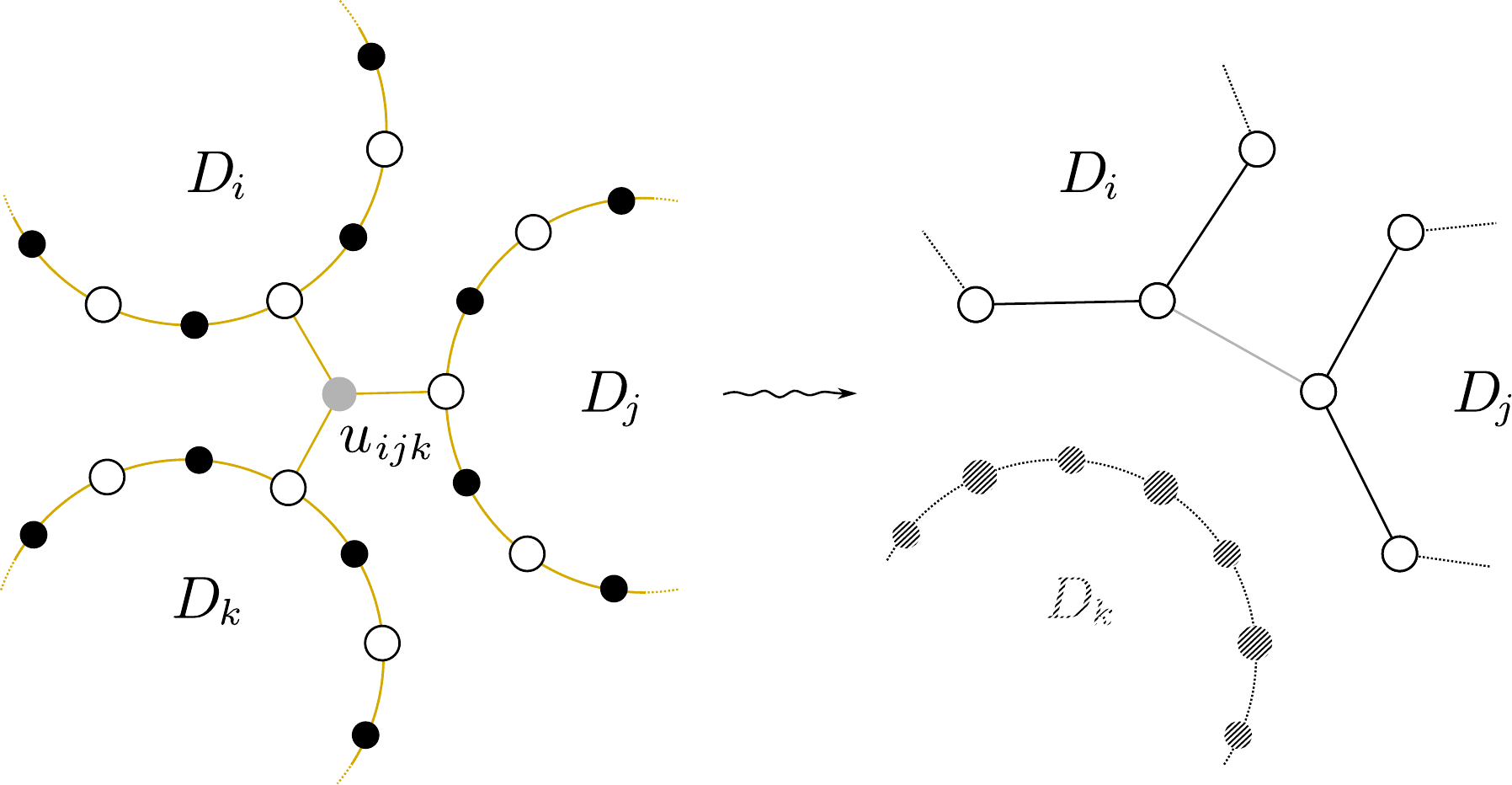}
	\caption{\label{fig:reductionA}%
		Three variable gadgets~$D_i$, $D_j$, $D_j$ connected by the gray
		clause vertex~$u_{ijk}$. Setting variable~$x_k$ to true and~$x_i$, $x_j$
		to false corresponds to the contractions on the right-hand side.
		The apex vertex~$a$ is not shown.
	}
\end{figure}

Let $\psi$ be a
\PosSAT instance with clauses $C_1,\dots C_m$ and variables
$x_1,\dots x_p$. We assume that every variable in $\psi$ appears in at
least 3 clauses; if not, we can duplicate clauses to achieve this without
changing the satisfiability of $\psi$. We construct a graph $G$ from~$\psi$
in the following manner (\cf Figure~\ref{fig:reductionA}):

\begin{enumerate}
	\item For each variable~$x_i$, create a cycle $D_i$ with as many
			  vertices as the frequency of~$x_i$ in $\psi$.
	\item Create an apex-vertex $a$ that is connected to every vertex of
				the cycles~$D_1,\ldots,D_p$.
	\item For each clause $\{x_i, x_j, x_k\}$, add a vertex~$u_{ijk}$
				to the graph and connect it to one vertex in $D_i,D_j,D_k$ each that
				has not yet been connected to any clause-vertex.
	\item Subdivide every edge appearing in the cycles~$D_1,\ldots,D_p$
			  and all edges incident to the apex~$a$.
\end{enumerate}

\noindent
For easier presentation, let us color the vertices of~$G$ as follows:
the vertices introduced in the first step are \emph{white}, the vertices
introduced in the third step \emph{gray}, and the subdivision vertices
created in the last step \emph{black} (the apex vertex~$a$ remains uncolored).
Note that the graph~$G$ is bipartite, where one side of the partition contains
exactly the white vertices and $a$. \looseness-1

Note that if the input formula~$\psi$ is planar, then the constructed graph
is planar and subcubic after removing the apex vertex~$a$.

\begin{lemma}\label{lem:nablahalf-sat-density}
	If $\psi$ is satisfiable then $G$ has a topological minor at depth $\half$
	of density $\frac{5m}{2m+1}$.
\end{lemma}
\begin{proof}
	We construct the minor $H$ by first smoothing each black vertex. Then, for
	each variable set to true, we delete the corresponding cycle $D_i$. Since
	$\psi$ is satisfiable and each clause has exactly one variable set to
	true, this step deletes exactly one neighbor from each gray vertex. We
	complete the construction of $H$ by smoothing out each gray vertex.

	$V(H)$ consists of exactly two vertices corresponding to each clause plus
	the apex $a$, for a total of $2m+1$ vertices. Since all vertices of $H$ were colored
	white in $G$, $a$ has degree $2m$. Aside from the edges incident to $a$,
	there are $m$ edges from smoothing gray vertices and $2m$ edges from
	smoothing black vertices, which yield a total of $5m$ edges. Thus, we have
	found a minor at depth $\half$ of density $\frac{5m}{2m+1}$.
\end{proof}

\begin{lemma}\label{lem:nablahalf-sat-correctness}
	If $G$ has a topological minor at depth $\half$ of density at least $\frac{5m}{2m+1}$,
	then the formula $\psi$ is satisfiable.
\end{lemma}
\begin{proof}
	Let $H$ be the densest shallow topological minor at depth $\half$ and
	fix some model of~$H$ in~$G$. We first argue
	that the nails of $H$ consist only of white vertex and potentially the apex
	vertex~$a$.

	\begin{claim}
		The nails of $H$ consist of the apex vertex~$a$ and some subset of white
		vertices.
 	\end{claim}
	\noindent
	First, since the density of $H$ is greater than two, its minimum degree is at
	least three (the removal of a degree-two vertex would increase the density).
	Since black vertices have degree two in $G$, the nails of the model forming $H$
	therefore cannot be black. Accordingly, every black vertex either does
	not participate in the formation of~$H$ or it is smoothed into an edge.

	Let us define~$G_b$ to be the graph obtained from~$G$ by smoothing all black
	vertices. Since black vertices have degree exactly two, this operation is uniquely
	defined. By the previous observation, $H$ can be obtained from~$G_b$ by only
	smoothing gray vertices and taking a subgraph. This of course implies that the
	nails of~$H$ are all either gray, white, or the apex vertex~$a$. Let us now exclude
	the first of these three cases: assume $y$ is a gray nail of~$H$ in~$G_b$. Again,
	the degree of~$y$ in~$H$ must be at least three to ensure maximal density of~$H$,
	and since~$y$ has degree three in~$G$ it must also have degree exactly three in~$H$.
	Note that the three neighbors of~$y$ are necessarily white and independent in~$G_b$,
	thus we can smooth~$y$ into an (arbitrary) edge between two of its neighbors.
	The newly obtained graph~$H'$ is again a half-shallow topological minor of~$G$ and
	it contains one vertex and two edges less than~$H$. Since the density of~$H$
	is greater than two, this implies that the density of~$H'$ is greater than that
	of~$H$, a contradiction. We conclude that the nails of~$H$ cannot be gray
	and therefore only consist of white vertices and, potentially, the apex vertex~$a$.
	To see that $a$ must be contained in $H$, simply note that otherwise the maximum
	degree of~$H$ would be three and as thus $H$s density would lie strictly below
	the assumed~$\frac{5m}{2m+1}$. In summary: $H$'s nails consist of the apex
	vertex $a$ and some subset of white vertices of~$G$, proving the claim.
	\smallskip

	Since the white vertices in $G$ are independent, the above claim further implies that
	the construction of~$H$ can be accomplished without smoothing white vertices. We can
	therefore divide said construction into two steps: first we smooth all gray
	and black vertices to construct a graph~$G_{gb}$ from~$G$ and then we take
	the subgraph~$H \subseteq G_{gb}$. In the following, we will refer to edges
	in $G_{gb}$ or~$H$ as \emph{gray} if they originated from smoothing a gray
	vertex and \emph{black} if they originated from smoothing a black vertex. Note that
	the set of black and gray edges partition~$E(G_{gb})$ and hence also~$E(H)$.

	We now denote by~$v_4$ the number of degree-four vertices in~$H$ and
	by~$v_3$ the number of degree-three vertices (as observed above, no vertex
	with degree lower than three can exist in $H$ and $a$ is the only vertex of degree greater than four). Since
	the number of gray edges is at most~$m$ and a degree-four vertex must be
	incident to a gray edge, we have that~$v_4 \leq 2m$. Let
  	$w = v_3 + v_4$ be the number of white vertices in~$H$ and~$\alpha = v_4/w$
  	the ratio of degree-four vertices among them. Using these quantities, we can express $H$'s density as
	\[
		\frac{2 v_3 + \fivehalf v_4}{v_3 + v_4 + 1}
		=  2\frac{w}{w+1} + \frac{\alpha}{2} \frac{w}{w+1}
		= \Big(2 + \frac{\alpha}{2} \Big) \frac{w}{w+1}
	\]
	which we combine with the density-requirement on~$H$ to obtain
	\begin{gather*}
	  \Big(2 + \frac{\alpha}{2} \Big) \frac{w}{w+1} \geq \frac{5m}{2m+1}
		\iff \alpha \geq 2\Big( \frac{5m}{2m+1} \frac{w+1}{w}\Big) - 4.
	\end{gather*}
	Note that the right-hand side is equal to one for~$w = 2m$, smaller than one
	for~$w > 2m$, and larger than one for~$w < 2m$. This last regime would imply
	the impossible~$\alpha > 1$ and we conclude $2m \leq w \leq 3m$, where the upper
	bound~$3m$ is simply the total number of white vertices in~$G$. Rewriting $w$
	as $\beta m$ for $2 \leq \beta \leq 3$, we revisit the density-constraint on~$H$:
	\begin{align*}
		\Big(2 + \frac{\alpha}{2} \Big) \frac{\beta m}{\beta m+1}  \geq \frac{5m}{2m+1}
		&\iff \big(2\beta + \frac{\alpha \beta}{2}\big) (2m+1) \geq 5 (\beta m+1) \\
		% &\iff (\alpha - 1) m + \big(2 + \frac{\alpha}{2}\big)  \geq \frac{5}{\beta}  \\
		&\implies (\alpha - 1) m + \frac{5}{2} \geq \frac{5}{\beta}.  \tag{$\star$}
	\end{align*}
	We will now show that $\alpha$, the fraction of degree-four vertices among all $w$ white
	vertices, needs to be one in order for~($\star$) to hold.
	To that end, we distinguish the following two cases:
	\begin{case}[$\beta = 2$]
		Assuming $\alpha \neq 1$, the largest possible value for~$\alpha$ is achieved
		when~$v_4 = 2m - 2$ (the case of exactly one gray edge missing from~$H$),
		resulting in~$\alpha = (2m-2) / 2m = 1 - \frac{1}{m}$.
		Plugging this value of~$\alpha$ and~$\beta = 2$ into~($\star$),
		we obtain that
		\[
			(1 - 1/m - 1) m + \frac{5}{2} = -1 + \frac{5}{2} \geq \frac{5}{2},
		\]
		a contradiction. Smaller values of~$\alpha$ lead to the same contradiction, and we conclude
		that necessarily $\alpha = 1$.
	\end{case}
	\begin{case}[$2 < \beta \leq 3$]
		Assuming $\alpha \neq 1$, the largest possible value for~$\alpha$ is achieved
		when~$v_4 = 2m$, resulting in $\alpha = 2m/\beta m = \frac{2}{\beta}$. Now ($\star$) becomes
		\[
			\Big(\frac{2}{\beta}-1\Big)m + \frac{5}{2} \geq \frac{5}{\beta}
			\iff m \leq \frac{10 - 5\beta}{2\beta} \cdot \frac{\beta}{2-\beta}
					= \frac{5}{2}.
		\]
		Thus for formulas~$\psi$ with at least three clauses, we arrive at a contradiction
		and conclude that~$\alpha = 1$.
	\end{case}

	\noindent
	We have now shown that a) $H$ contains \emph{only} vertices of degree four and b)
	that $|H| \geq 2m$. Since there cannot be more than~$2m$ vertices of degree four,
	we conclude that~$H$ has exactly $2m$ vertices. Note that therefore $H$ must consist
	of a collection of black-edge cycles with a total of~$2m$ vertices, each of which is incident
	to exactly one of the $m$ gray edges. Note that each black cycle~$B_i$ in~$H$ corresponds to a
	cycle~$D_i$ (associated with variable~$x_i$) in~$G$, where~$B_i$ was constructed from~$D_i$
	by smoothing black vertices. Thus we can associate every black cycle~$B_i$ in~$H$ with a
	variable~$x_i$ in $\psi$. We claim that setting all such variables~$x_i$ that have a black
	cycle~$C_i$ in $H$ to false and all other variables to true is a 1-in-3 satisfying assignment
	of~$\psi$.

	Consider any clause~$\{x_i, x_j, x_k\}$ in~$\psi$. The corresponding gray vertex~$u_{ijk}$
	in~$G$ was smoothed into a gray edge~$e_{ijk}$ in~$H$, since all~$m$ gray are present in~$H$.
	Accordingly, exactly two of the three black cycles~$B_i,B_j,B_k$ are contained in~$H$. Thus the
	assignment constructed above will set exactly two of the variables~$x_i,x_j,x_k$ to false and one
	variable to true. This argument holds for every clause in~$\psi$ and we conclude that the constructed
	assignment is 1-in-3 satisfying, proving the lemma.
\end{proof}

\noindent
This concludes the reduction from \PosSAT. Note that an optimal solution in
the reduction necessarily does not use any edges from the original graph, but
only edges resulting from contractions. Therefore the reduction works for
both~\Problem{Dense \half-STM} and~\Problem{Dense \half-SD}.
As noted above, for~$r = 1$ the constructed graph is bipartite. Since the
latter set has degree at most three, Theorem~\ref{thm:dense-bipartite-subdiv}
follows.

In order for the reduction to work for arbitrary odd~$r$, we need to modify
the construction in two places: first, we subdivide every edge in the
clause gadget~$(r-1)/2$ times. Second, instead of subdividing all edges
appearing in the cycles~$D_1,\ldots,D_p$ and edges incident to the apex~$a$
once, we subdivide them~$r$ times. The correctness of this reduction follows
from easy modifications to Lemma~\ref{lem:nablahalf-sat-density}
and~\ref{lem:nablahalf-sat-correctness}, concluding our proof of Theorem~\ref{thm:np-hardness}
for odd values of~$r$. Finally, to see that
the above reduction also proofs Theorem~\ref{thm:eth-lowerbound} for odd~$r$,
simply note that the reduction results in a graph of size~$\Theta(m)$ and the
ETH lower bound follows from Proposition~\ref{prop:possat-eth}.

%   88888888b
%   88
%  a88aaaa    dP   .dP .d8888b. 88d888b.
%   88        88   d8' 88ooood8 88'  `88
%   88        88 .88'  88.  ... 88    88
%   88888888P 8888P'   `88888P' dP    dP
%
%
\paragraph*{Reduction for $r$ even}

\noindent
Let $\psi$ be a \PosSAT instance as described above.
Construct graph $G$ in the following manner. We once again create a cycle
$D_i$ for each variable $x_i$, connect an apex vertex $a$ to each vertex on
the cycles, and color these vertices white. For this construction, however, we
subdivide all edges between white vertices twice i.e. each white-white edge is
replaced by a three-edge path. As with our previous construction, the
subdivision vertices are all colored black. For each clause $C_i = \{x_j, x_k, x_\ell\}$,
we add a triangle $u_{ij}, u_{ij}, u_{ik}$ and
connect it to the vertices from $D_j$, $D_k$, and $D_\ell$ corresponding to
$C_i$ such that $u_{ij}$ is incident to the vertex from $D_j$ etc. We color
each of these vertices gray.

\begin{lemma}\label{lem:nablaone-sat-density}
    If $\psi$ is satisfiable then $G$ has a topological minor at depth~$1$
    of density $\frac{5m}{2m+1}$.
\end{lemma}
\begin{proof}
    We construct the minor $H$ by first smoothing each black vertex. Then, for
    each variable set to true, we delete the corresponding cycle $D_i$. Since
    $\psi$ is satisfiable and each clause has exactly one variable set to
    true, each gray triangle has two vertices of degree three and one of degree
    two. The degree two gray vertices are deleted, leaving the remaining gray
    vertices to lie on three-edge paths between white vertices. These paths are
    subsequently smoothed to create white-white edges.

    $V(H)$ consists of exactly two vertices corresponding to each clause plus
    $a$, for a total of $2m+1$ vertices. Since all vertices of $H$ were colored
    white in $G$, $a$ has degree $2m$. Aside from the edges incident to $a$,
    there are $m$ edges from smoothing gray vertices and $2m$ edges from
    smoothing black vertices, which yield a total of $5m$ edges. Thus, we have
    found a minor at depth $1$ of density $\frac{5m}{2m+1}$.
\end{proof}

\begin{lemma}\label{lem:nablaone-sat-correctness}
    If $G$ has a topological minor at depth $1$ of density $\frac{5m}{2m+1}$ then
    $\psi$ is satisfiable.
\end{lemma}
\begin{proof}
    Let $H'$ be the densest topological minor at depth $1$. For the same reasons
    presented in Lemma~\ref{lem:nablahalf-sat-correctness}, $H'$ has no black nails, and
    thus we can smooth all black vertices into white-white edges. This lack of
    black nails also implies that no white vertices can be smoothed to form a
    new edge incident to a gray vertex.

    If $H'$ contains all the gray and white vertices, it has $3m$ degree 4 white
    vertices, $3m$ degree 3 gray vertices, and $a$ with degree $3m$ for a total
    of $12m$ edges and $6m+1$ vertices. This implies a density below
    $\frac{5m}{2m+1}$, and thus not all white and gray vertices are nails.

    Since the gray vertices induce triangles, there is no way to smooth gray
    vertices to create a new gray-gray edge. Consider one such triangle $T_i$.
    If we smooth two vertices in $T_i$ to create a single gray-white edge, the
    gray nail has degree 2 and should be deleted instead to increase the
    density. On the other hand, smoothing exactly one gray vertex to create a
    gray-white edge cause the remaining gray vertex to have degree two. Thus,
    any gray nail in $H'$ is adjacent to three white vertices. Note that instead
    of having a gray nail, we could delete one gray vertex and smooth the other
    two into a white-white edge. The proof in Lemma~\ref{lem:nablahalf-sat-correctness}
    already demonstrated that forming the white-white edges is necessary to
    yield a density of $\frac{5m}{2m+1}$, and thus $H'$ has no gray nails.

    Since the gray vertices must be smoothed and deleted to create two degree 4
    vertices and one degree 3 vertex per clause, the arguments in
    Lemma~\ref{lem:nablahalf-sat-correctness} imply that for $H'$ to have density
    $\frac{5m}{2m+1}$, $\psi$ must be satisfiable.
\end{proof}

\noindent
In order for the reduction to work for arbitrary even~$r$, we again modify
the construction in two places: first, we subdivide every edge of the
triangle making up the clause gadget~$r/2-1$ times. Second, instead of subdividing all edges
appearing in the cycles~$D_1,\ldots,D_p$ and edges incident to the apex~$a$
twice, we subdivide them~$r$ times.
With both cases of~$r$ even or odd covered, we conclude that
Theorem~\ref{thm:np-hardness} and Theorem~\ref{thm:eth-lowerbound} hold true.

\subsection{Excluding a $2^{o(\tw^2)} n^{O(1)}$-algorithm}\label{sec:tw-lower-bound}
We show in this section that the ETH implies that we cannot get a
single-exponential algorithm parameterized by treewidth
for~\Problem{Dense $r/2$-STM} for~$r \geq 2$.

\begin{theorem}\label{thm:SETH_lower_bound}
  Unless the ETH fails, there is no algorithm that decides \Problem{Dense 1-Shallow
  Topological Minor} on a graph with treewidth $t$ in time~$2^{o(t^2)} n^{O(1)}$.
\end{theorem}
\noindent Our proof proceeds via a reduction from \Problem{CNF-SAT}. Assume
that the CNF formula $\Phi$ with variables $x_1,\dots, x_n$ and clauses $C_1,
\dots, C_m$ is such that $\sqrt{n}$ is an even integer; if not, we pad $\Phi$
with dummy variables that appear in no clauses, which does not affect the
answer to $\Phi$. Figure~\ref{fig:reductionB} contains a sketch of the construction
outlined in the following.

\begin{figure}
  \centering\includegraphics[width=.7\textwidth]{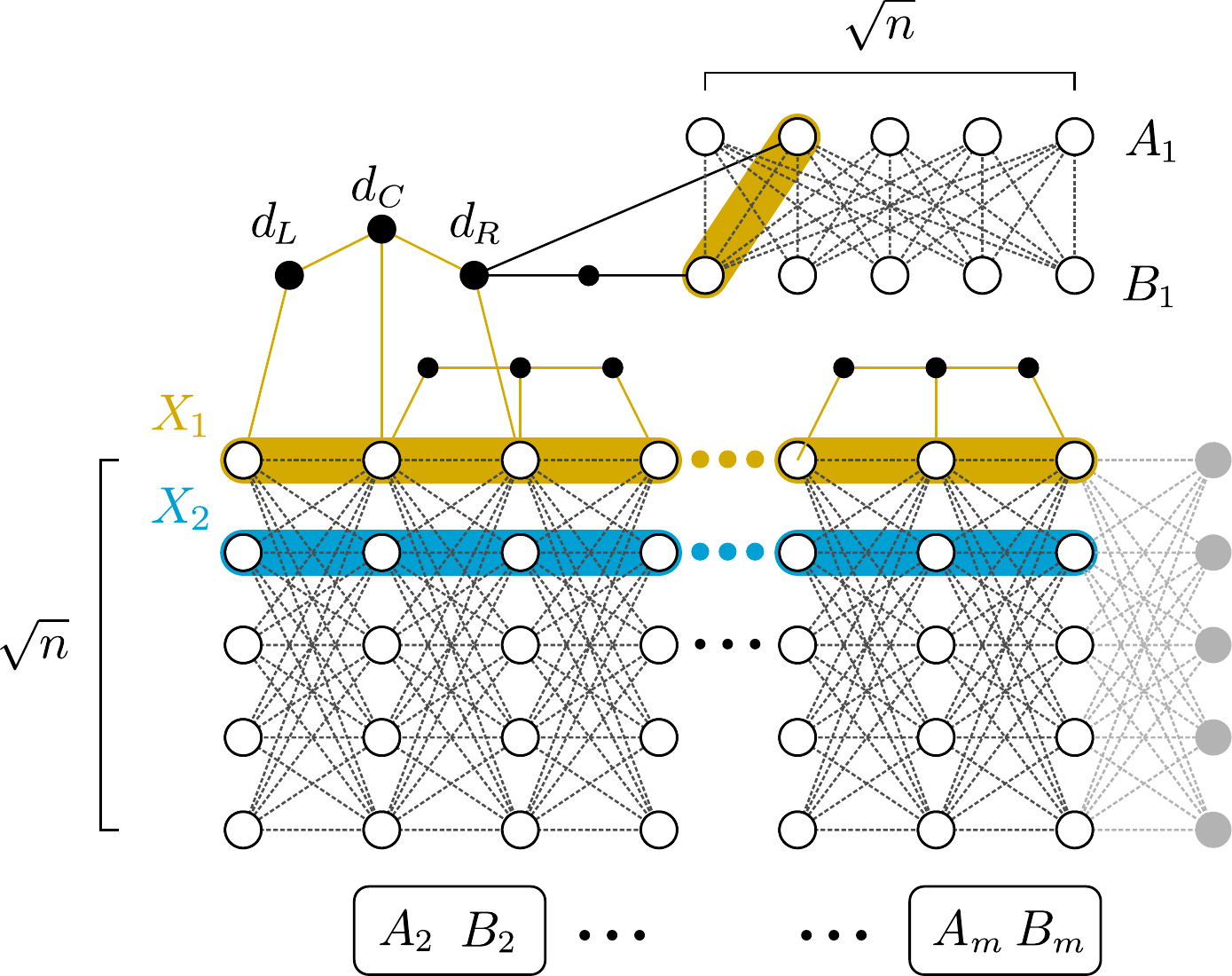}
  \caption{\label{fig:reductionB}%
    A sketch of the construction for Theorem~\ref{thm:SETH_lower_bound}, with an exemplary
    connection of the variable-path~$X_1$ to the first clause gadget (here, $x_1$ appears
    negatively in~$C_1$). Dashed edges denote parts that are actually connected via
    decision gadgets. $3$-paths between the grid~$R$ and the clause gadgets~$(A_i,B_i)$ are not drawn.
  }
\end{figure}

\smallskip\noindent\emph{Decision gadget}:
The reduction will use sequences of vertices connected by decision gadgets.
The decision gadget is a path of three vertices $d_L,d_C,d_R$ which we will always connect to a sequence of three vertices.
For a sequence of vertices $v_1,v_2,v_3$, connecting the decision gadget to the sequence involves adding the edges $\{d_L,v_1\}$,
$\{d_C,v_2\}$, and $\{d_R,v_3\}$.

\smallskip\noindent\emph{Variable gadgets}:
We construct a grid of vertices $R$ with $\sqrt{n}$ rows and $m$ columns,
denoting with $R[i,j]$ the vertex in the $i$th row and $j$th column. Each
variable $x_i$ will be represented by a sequence of $m$ vertices $X_i$, one
from each column.  We will denote with $X_i[j]$ the $j$th vertex in the
sequence $X_i$ for any $1\leq j\leq m$.  In order to represent each of the $n$ variables with $m$ vertices using a $\sqrt{n} \times m$ grid, the sequences $X_i$ must overlap and each vertex in $R$ is part of the representation of $\sqrt n$ varibles.
Specifically, let
$$X_i[j] = R[(i+j\lfloor i/\sqrt{n}\rfloor)\bmod\sqrt{n} , j].$$
In other words, two vertices in successive columns will be in the same row in sequences $X_1, \dots , X_{\sqrt{n}}$;
they will be one row apart (``wrapping around'' to the top from the bottom) in sequences $X_{1+\sqrt{n}}, \dots, X_{2\sqrt{n}}$,
two rows apart in $X_{1+2\sqrt{n}},\dots, X_{3\sqrt{n}}$, and so on.

For each sequence $X_i$, we connect $X_i[j-1]$, $X_i[j]$, and $X_i[j+1]$ to a
decision gadget. Denote such a decision gadget as $D_{i,j}$. We also ``wrap
around'' $X_i$ by connecting $X_i[m-1],X_i[m], X_i[1]$ and $X_i[m],X_i[1],
X_i[2]$ to their own decision gadgets.

\smallskip\noindent\emph{Clause gadgets}:
Each clause $C_i$ will be represented by a
bipartition of vertices $A_i,B_i$ where $|A_i|=|B_i| = \sqrt{n}$. Let
$A_i[j]$ be the $j$th vertex in $A_i$ and $B_i[j]$ likewise. Let $\sigma$ be an
ordering of the vertices in $A_i\cup B_i$ corresponding to an Eulerian tour of
a biclique with bipartition $A_i, B_i$. Assume without loss of
generality that $\sigma_i = A_i[\sqrt{n}],
B_i[1],\dots, B_i[\sqrt n], A_i[\sqrt{n}]$ and note that every pair of
vertices $a\in A_i$ and $b\in B_i$ appears consecutively exactly once in
$\sigma$. For each consecutive triple of vertices in $\sigma_i$, attach a
decision gadget (but do not ``wrap around''). \looseness-1

%Let $\mathcal A = \bigcup_{i\leq m} A_i$ and $\mathcal B$ likewise.
%Create a set of vertices $Y$ of size $2\sqrt n-1$ and connect via a 3-edge path each vertex in $Y$ to each vertex in $\mathcal A\cup \mathcal B$.
%Add an additional $\sqrt{n}$ vertices $Y'$ connected via 3-edge paths only to each $A_i[\sqrt{n}+1]$ and $B_i[\sqrt{n}+1]$.

\smallskip\noindent\emph{Connecting variables and clauses}:
For each pair of vertices $(A_i[j], B_i[k])$ for $1\leq j,k\leq \sqrt{n}$
assign the pair with a unique variable $x_\ell$. Connect $X_\ell[i]$ to
$A_i[j]$ and $B_i[k]$ via 3-edge paths. If $x_\ell$ appears in clause $C_i$
positively, connect $d_L$ of the decision gadget $D_{i,\ell}$ to $B_i[k]$ via
an edge and to $A_i[j]$ via a 2-edge path. If it appears negatively, add the
same connections to $d_R$ of $D_{i,\ell}$ instead.
%Finally, connect two new vertices $a,a'$ to all vertices in $\mathcal A \cup \mathcal B \cup R$ via 3-edge paths.

\medskip\noindent
With the description of the reduction completed, let us now proof its correctness,
\ie We prove that $G$ has a 1-STM of density
%$\rho = \frac{ 8mn+12m\sqrt{n}+3m}{3m\sqrt{n}+2m+3\sqrt{n}+1}$
$\rho = \frac{4\sqrt{n}}{3}$ if and only if $\Phi$ is satisfiable.

\noindent\emph{Forward direction}:
To prove the forward direction, we show how the satisfying assignment yields a
topological minor of the desired density. We note that a cyclical sequence of
vertices joined by decision gadgets can form a cycle in one of two ways: by
smoothing each $d_C$ and $d_L$ or each $d_C$ and $d_R$. Let the former be know
as the \emph{left configuration} of those sequenced gadgets and the latter the
\emph{right configuration}. Create a cycle on the vertices in $X_i$ by
choosing the right configuration if $x_i$ is true and the left configuration
if $x_i$ is false.

For each clause, pick an arbitrary variable $x_\ell$ that satisfies it and let
$a$ and $b$ be the pair of vertices from $A_i$ and $B_i$ assigned to that
variable. If $a <_\sigma b$, set all decision gadgets preceding $a$ in
$\sigma$ to the left configuration and all the decision gadgets succeeding $b$
to the right configuration; do the reverse if $b<_\sigma a$. Thus $A_i$ and
$B_i$ form a biclique missing the $ab$ edge. Since there is a 3-edge path from
$a$ to $b$ in $G$ through a vertex in $D_{i,\ell}$ that has not been smoothed,
we can use that path to form the $ab$ edge. Smooth the remaining 3-edge
induced paths in $G$ and delete the vertices from the decision gadgets that
were not contracted.

The nails of the resulting minor are exactly $\mathcal A\cup \mathcal B\cup
R$. There are $m\sqrt{n}$ vertices in $R$ and each one participates in
$\sqrt{n}$ variable gadgets. Since each variable gadget becomes a cycle there
are $mn$ edges within $R$. The $m$ clause gadgets become bicliques on
$2\sqrt{n}$ vertices each and thus contain $mn$ edges in total. Each variable
gadget ends up with two edges into each clause gadget, for a total of $2mn$
edges connecting them. In total, this makes $4mn$ edges and $3m\sqrt{n}$
vertices, exactly $\rho$.

\smallskip\noindent\emph{Reverse direction}:
We now prove the reverse direction by assuming $\Phi$ is unsatisfiable.
Let $H'$ be a 1-STM with density $\rho$.
% Let $H(\Phi)$ be the minor constructed by the process described in the
% forward direction proof for an arbitrary truth assignment to $\Phi$. Because
% there is no decision gadget that wraps around the order $\sigma_i$ in $C_i$,
% setting the last decision gadget in $C_i$ to the left configuration fails to
% create the edge between $B_i[\sqrt n]$ and $A_i[\sqrt n]$, while setting the
% first gadget to the right configuration fails to create the edge between
% $A_i[\sqrt n]$ and $B_i[1]$. In addition, if a decision gadget set the left
% configuration is immediately followed by a right configuration, that edge
% will not be formed; thus, there is no way to contract the decision gadgets
% to create a full biclique in the clause gadget. Because not every clause is
% satisfied, there will be a clause gadget that is unable to contract the
% missing edge to complete its biclique and $H(\Phi)$ will not be a minor of
% density $\rho$. Thus we can assume that $H'$ is formed via a different
% construction.

Since $\rho$ is $\Theta(\sqrt{n})$ and the vertices in the induced paths and
decision gadgets have degree at most 4, we can assume that none of those
vertices appears as a nail in $H'$. Thus, the nail set of $H'$ is a subset of
$\mathcal A\cup \mathcal B\cup R$. The only paths between these nail
candidates that use at most three edges do not contain nail candidates as
interior vertices, meaning nail candidates are never smoothed. Let $H(\Phi)$
be the minor constructed by the process described in the forward direction
proof for an arbitrary satisfying assignment of~$\Phi$. Observe that for fixed
$n$ and $m$, the minor $H(\psi)$ is identical for every satisfiable formula
$\psi$; let $\mathcal{H}({n,m})$ be that minor. Moreover, every pair of nail
candidates that has a 3-edge path between them in $G$ is adjacent in
$\mathcal{H}({n,m})$, meaning that $H'$ is a either
a subgraph of or identical to $\mathcal{H}(n,m)$.

We now show that no proper induced subgraph of $\mathcal H(n,m)$ has density
$\rho$. If a graph is $d$-regular and connected, it has edge density $d/2$;
deleting part of a degree regular graph leaves vertices with degree less than
$d$, so no proper subgraph reaches that density. Therefore $R$ and each
$A_i\cup B_i$ achieve their maximum densities of $\sqrt{n}$ and $\sqrt{n}/2$
only when including the entire subgraph, implying a dense subgraph must
contain portions of both vertex and clause gadgets. The only vehicle for
increasing density is to use edges between vertex and clause gadgets, which
means we should only include a vertex in $R$ in a subgraph if it also contains
its neighbors from $\mathcal A$ and $\mathcal B$. Let $G'$ be the subgraph of
$G$ induced on an $i\times j$ subgrid of $R$ and all of its neighbors in
$\mathcal A\cup \mathcal B$. The density within the subset of $\mathcal A \cup
\mathcal B$ is greatest when those vertices induce $j\times j$ bicliques, so
we assume they do. Under this assumption, the density of $G'$ is
$\frac{4i}{3}$ if $j=m$ and $\frac{(4j-1)i}{3j}$ if $j<m$ since the edges that
wrap around $R$ cannot be realized. In either case, the density is strictly
less than $\rho$ unless $j=m$ and $i=\sqrt n$ i.e. exactly $\mathcal H(n,m)$.

A decision gadget connected to sequential vertices $v_1, v_2, v_3$ can only
create the edge $v_1,v_2$ or the edge $v_2,v_3$, since $d_C$ needs to be
smoothed to construct either edge. Consequently, choosing to set some decision
gadgets in the same variable gadget to opposite configurations (or neither
configuration) creates at least one fewer edge than if they were all set to
the same decision. Thus, the configurations of the variable gadgets correspond
to some truth assignment to $\Phi$ in the way intended in $H(\Phi)$. This
indicates that there is a clause gadget that has no neighbors in a variable
gadget that can be used to be smoothed into an edge in the clause gadget.
However, because there is one fewer decision gadget than the number of
biclique edges in the clause gadgets of $\mathcal H(n,m)$, $H'$ cannot realize
all possible edges in the biclique and thus there is no 1-STM of density
$\rho$.

\medskip
\noindent It stands to prove that the above reduction has the proper
implications for a parameterization by treewidth. Using cops and robbers,
we can show that $G$ has treewidth $O(\sqrt n)$ as follows: We
permanently station $\sqrt n$ cops on the first column of $R$. We use
$2\sqrt{n}$ cops to walk through the columns of $R$ sequentially; when a
column is completely covered with cops we can explore its corresponding clause
gadget with a separate unit of $2\sqrt{n}$ cops. In total, this requires
$O(\sqrt{n})$ cops.
Although the formula $\Phi$ may have been padded with additional variables, it
would only have been enough to increase $\sqrt n$ by 2. This means the number
of variables in the unpadded instance is still $\Theta(n)$. Thus, if an
algorithm parameterized by treewidth $t$ could find a dense 1-STM in time
$2^{o(t^2)}n^{O(1)}$, then we could use our reduction to solve \Problem{CNF-SAT} in
time $2^{o(n)}n^{O(1)}$, violating ETH. This concludes the proof of Theorem~\ref{thm:SETH_lower_bound}.

An immediate consequence is that, unlike
\Problem{Dense \half-STM}, \Problem{Dense $1$-STM} is still NP-hard when the exact nail set
is known.

\section{Conclusion}
\noindent
We showed that finding dense substructures that are just slightly less local than
subgraphs is computationally hard, and even a parameterization by treewidth cannot
provide very efficient algorithms. While our first reduction excludes a subexponential
exact algorithm assuming the ETH, we could not exclude an algorithm with a running time
of~$(2-\eps)^n n^{O(1)}$. Is such an algorithm possible for~$r=1$, or can one find a
tighter reduction that provides a corresponding SETH lower bound?
Our second reduction rules out a $2^{o(\tw^2)} n^{O(1)}$-algorithm
for~$r = 2$. Is a faster algorithm for~$r=1$ possible?

Finally, we ask whether there is a sensible notion of substructures that fit in between
\half-shallow topological minors and subgraphs for which we can find the
densest occurrence in polynomial time.

\noindent\textbf{Acknowledgements}.
{\small This work supported in part by the DARPA GRAPHS Program and the
Gordon \& Betty Moore Foundation's Data-Driven Discovery Initiative through
Grants SPAWAR-N66001-14-1-4063 and GBMF4560 to Blair D.~Sullivan.}

\bibliographystyle{plain}
\bibliography{biblio}

\end{document}